\documentclass[a4paper,UKenglish,cleveref, autoref, thm-restate]{lipics-v2021}

\usepackage{dsfont}
\usepackage{mathtools}
\usepackage{pgfplots}
\usepackage{pgfplots}
\pgfplotsset{compat=newest}

\usepgfplotslibrary{groupplots}
\usetikzlibrary{pgfplots.polar}
\usetikzlibrary{pgfplots.statistics}
\usetikzlibrary{arrows.meta}
\usetikzlibrary{overlay-beamer-styles}

\definecolor{veryLightGrey}{HTML}{F2F2F2}
\definecolor{lightGrey}{HTML}{DDDDDD}
\definecolor{colorSimdRecSplit}{HTML}{444444}
\definecolor{colorChd}{HTML}{377EB8}
\definecolor{colorRustFmph}{HTML}{A65628}
\definecolor{colorRustFmphGo}{HTML}{A65628}
\definecolor{colorSicHash}{HTML}{4DAF4A}
\definecolor{colorPthash}{HTML}{984EA3}
\definecolor{colorDensePtHash}{HTML}{377EB8}
\definecolor{colorRecSplit}{HTML}{FF7F00}
\definecolor{colorBbhash}{HTML}{F781BF}
\definecolor{colorShockHash}{HTML}{F8BA01}
\definecolor{colorBipartiteShockHash}{HTML}{F781BF}
\definecolor{colorBipartiteShockHashFlat}{HTML}{984EA3}
\definecolor{colorBmz}{HTML}{000000}
\definecolor{colorBdz}{HTML}{444444}
\definecolor{colorFch}{HTML}{444444}
\definecolor{colorChm}{HTML}{A65628}
\definecolor{colorFiPS}{HTML}{FF7F00}
\definecolor{colorConsensus}{HTML}{4DAF4A}
\definecolor{colorPtrHash}{HTML}{F8BA01}
\definecolor{colorMorphisHash}{HTML}{A65628}
\definecolor{colorMorphisHashFlat}{HTML}{E41A1C}

\colorlet{colorBruteForce}{colorSicHash}
\colorlet{colorRotationFitting}{colorChd}
\definecolor{colorCuckoo}{HTML}{E41A1C}

\pgfdeclareimage[interpolate,height=2.0mm,width=2.0mm]{shockhashMark}{fig/shockhash}
\pgfdeclareplotmark{shockhash}{\pgftext[at=\pgfpointorigin]{\pgfuseimage{shockhashMark}}}
\pgfdeclareimage[interpolate,height=1.85mm,width=1.85mm]{shockhashInverseMark}{fig/shockhashInverse}
\pgfdeclareplotmark{shockhashInverse}{\pgftext[at=\pgfpointorigin]{\pgfuseimage{shockhashInverseMark}}}
\pgfdeclareimage[interpolate,height=1.85mm,width=1.85mm]{phobicMark}{fig/phobic}
\pgfdeclareplotmark{phobic}{\pgftext[at=\pgfpointorigin]{\pgfuseimage{phobicMark}}}
\pgfdeclareimage[interpolate,height=1.85mm,width=1.85mm]{phobicInverseMark}{fig/phobicInverse}
\pgfdeclareplotmark{phobicInverse}{\pgftext[at=\pgfpointorigin]{\pgfuseimage{phobicInverseMark}}}
\pgfdeclareimage[interpolate,height=1.85mm,width=1.85mm]{fingerprintMark}{fig/hand}
\pgfdeclareplotmark{fingerprint}{\pgftext[at=\pgfpointorigin]{\pgfuseimage{fingerprintMark}}}
\pgfdeclareimage[interpolate,height=1.85mm,width=1.85mm]{fingerprintInverseMark}{fig/handInverse}
\pgfdeclareplotmark{fingerprintInverse}{\pgftext[at=\pgfpointorigin]{\pgfuseimage{fingerprintInverseMark}}}
\pgfdeclareimage[interpolate,height=2.0mm,width=2.0mm]{consensusMark}{fig/consensus}
\pgfdeclareplotmark{consensus}{\pgftext[at=\pgfpointorigin]{\pgfuseimage{consensusMark}}}
\pgfdeclareimage[interpolate,height=2.0mm,width=2.0mm]{consensusInverseMark}{fig/consensusInverse}
\pgfdeclareplotmark{consensusInverse}{\pgftext[at=\pgfpointorigin]{\pgfuseimage{consensusInverseMark}}}

\pgfplotscreateplotcyclelist{myColorList}{%
  colorShockHash,mark=triangle\\
  colorBdz,mark=pentagon\\%
  colorChd,mark=square\\%
  colorSicHash,mark=o\\%
  colorRecSplit,mark=diamond\\%
  colorPthash,mark=x\\%
  colorBbhash,mark=flippedTriangle\\%
}
\pgfplotscreateplotcyclelist{compCycle}{%
  colorShockHash
}
\pgfdeclareplotmark{flippedTriangle}{%
  \pgfpathmoveto{\pgfqpointpolar{-90}{1.2\pgfplotmarksize}}%
  \pgfpathlineto{\pgfqpointpolar{30}{1.2\pgfplotmarksize}}%
  \pgfpathlineto{\pgfqpointpolar{150}{1.2\pgfplotmarksize}}%
  \pgfpathclose%
  \pgfusepath{stroke}
}
\pgfdeclareplotmark{rightTriangle}{%
  \pgfpathmoveto{\pgfqpointpolar{0}{1.2\pgfplotmarksize}}%
  \pgfpathlineto{\pgfqpointpolar{120}{1.2\pgfplotmarksize}}%
  \pgfpathlineto{\pgfqpointpolar{240}{1.2\pgfplotmarksize}}%
  \pgfpathclose%
  \pgfusepath{stroke}%
}
\pgfdeclareplotmark{leftTriangle}{%
  \pgfpathmoveto{\pgfqpointpolar{-180}{1.2\pgfplotmarksize}}%
  \pgfpathlineto{\pgfqpointpolar{-60}{1.2\pgfplotmarksize}}%
  \pgfpathlineto{\pgfqpointpolar{60}{1.2\pgfplotmarksize}}%
  \pgfpathclose%
  \pgfusepath{stroke}%
}

\definecolor{veryLightGrey}{HTML}{DDDDDD}
\definecolor{greenToPurple1}{HTML}{F9BB01}
\definecolor{greenToPurple2}{HTML}{CC8417}
\definecolor{greenToPurple3}{HTML}{A36029}
\definecolor{greenToPurple4}{HTML}{7F4C37}
\definecolor{greenToPurple5}{HTML}{5F4440}
\definecolor{greenToPurple6}{HTML}{444444}
\pgfplotscreateplotcyclelist{plotSicHashBucketSizeList}{%
    greenToPurple1,mark=o\\%
    greenToPurple2,mark=triangle\\%
    greenToPurple3,mark=pentagon\\%
    greenToPurple4,mark=square\\%
    greenToPurple5,mark=x\\%
    greenToPurple6,mark=diamond\\%
}

\pgfplotsset{
  mark repeat*/.style={
    scatter,
    scatter src=x,
    scatter/@pre marker code/.code={
      \pgfmathtruncatemacro\usemark{
        or(mod(\coordindex,#1)==0, (\coordindex==(\numcoords-1))
      }
      \ifnum\usemark=0
        \pgfplotsset{mark=none}
      \fi
    },
    scatter/@post marker code/.code={}
  },
  log x ticks with fixed point/.style={
      xticklabel={
        \pgfkeys{/pgf/fpu=true}
        \pgfmathparse{exp(\tick)}%
        \pgfmathprintnumber[fixed relative, precision=3]{\pgfmathresult}
        \pgfkeys{/pgf/fpu=false}
      }
  },
  log y ticks with fixed point/.style={
      yticklabel={
        \pgfkeys{/pgf/fpu=true}
        \pgfmathparse{exp(\tick)}%
        \pgfmathprintnumber[fixed relative, precision=3]{\pgfmathresult}
        \pgfkeys{/pgf/fpu=false}
      }
  },
  every axis/.style={scale only axis},
  major grid style={thin,dotted},
  minor grid style={thin,dotted},
  ymajorgrids,
  yminorgrids,
  every axis/.append style={
    line width=0.9pt,
    tick style={
      line cap=round,
      thin,
      major tick length=4pt,
      minor tick length=2pt,
    },
    mark options={solid},
  },
  legend cell align=left,
  legend style={
    line width=0.7pt,
    /tikz/every even column/.append style={column sep=2mm,black},
    /tikz/every odd column/.append style={black},
    mark options={solid},
    font=\small,
  },
  title style={yshift=-2pt},
  enlarge x limits=0.04,
  scale only axis,
  /pgf/number format/1000 sep={},
  axis lines*=left,
  xlabel near ticks,
  ylabel near ticks,
  axis lines*=left,
  label style={font=\footnotesize},
  every axis y label/.append style={yshift=-1pt,inner sep=0,outer sep=0},
  tick label style={font=\footnotesize},
  cycle list name=myColorList,
  plotEvalPareto/.style={
    width=55mm,
    height=40mm,
    xmax=2,
  },
}

\usepackage{caption}
\usepackage{subcaption}
\usepackage{hyperref}
\usepackage{booktabs}
\usepackage{listings}
\usepackage{algorithm}
\usepackage{csquotes}
\usepackage[noend]{algpseudocode}
\usepackage{placeins}
\usepackage{bbding}
\usepackage{multicol}
\usepackage{verbatim}
\usepackage[pscoord]{eso-pic}

\hideLIPIcs

\usepackage{xspace}
\usepackage[shortlabels]{enumitem}

\def\consensus{\texorpdfstring{C\scalebox{0.8}{ONSENSUS}}{CONSENSUS}\xspace}

\newcommand{\mytitlerunning}{MorphisHash: Improving Space Efficiency of ShockHash for Minimal Perfect Hashing}
\title{\texorpdfstring{\mytitlerunning}{\mytitlerunning}}
\titlerunning{\mytitlerunning}

\author{Stefan Hermann}{Karlsruhe Institute of Technology, Germany}{hermann@kit.edu}{https://orcid.org/0000-0001-9183-2926}{}

\newcommand{\myauthorrunning}{S. Hermann}
\authorrunning{\myauthorrunning}
\Copyright{Stefan Hermann}
\hypersetup{
	colorlinks=true,
	pdftitle={\mytitlerunning},
	pdfauthor={\myauthorrunning},
	pdfsubject={}
}

\EventEditors{Anne Benoit, Haim Kaplan, Sebastian Wild, and Grzegorz Herman}
\EventNoEds{4}
\EventLongTitle{33rd Annual European Symposium on Algorithms (ESA 2025)}
\EventShortTitle{ESA 2025}
\EventAcronym{ESA}
\EventYear{2025}
\EventDate{September 15--17, 2025}
\EventLocation{Warsaw, Poland}
\EventLogo{}
\SeriesVolume{351}
\ArticleNo{7}

\DeclareMathOperator{\defe}{def}
\DeclareMathOperator{\rank}{rank}

\ccsdesc[500]{Theory of computation~Data compression}
\ccsdesc[500]{Information systems~Point lookups}
\ccsdesc[500]{Theory of computation~Bloom filters and hashing}
\ccsdesc[500]{Mathematics of computing~Random graphs}

\keywords{compressed data structure, perfect hashing, random graph, pseudoforest, component} %

\category{} %
\supplementdetails{Software}{https://github.com/stefanfred/MorphisHash}

\funding{
	This work was supported by funding from the pilot program Core Informatics at KIT (KiKIT) of the Helmholtz Association (HGF).
}

\acknowledgements{I thank Stefan Walzer for proofreading an earlier version of this paper.}

\newcommand{\myparagraph}[1]{\subparagraph{#1\@}}

\nolinenumbers

\newcommand{\fvref}[1]{in \cref{#1}}
	
\begin{document}
\maketitle

\def\speedup{21}

\begin{abstract}
A minimal perfect hash function (MPHF) maps a set of $n$ keys to unique positions $\{1, \ldots, n\}$.
Representing an MPHF requires at least $\log_2(e)\approx 1.443$ bits per key.
ShockHash is a technique to construct an MPHF and requires just slightly more space.
It gives each key two random candidate positions.
If each key can be mapped to one of its two candidate positions such that there is exactly one key mapped to each position, then an MPHF is found.
If not, ShockHash repeats the process with a new set of random candidate positions.
ShockHash has to store how many repetitions were required and for each key to which of the two candidate positions it is mapped.
However, when a given set of candidate positions can be used as MPHF then there is not only one but multiple ways of mapping the keys to one of their candidate positions such that the mapping results in an MPHF.
This redundancy makes up for the majority of the remaining space overhead in ShockHash.
In this paper, we present MorphisHash which almost completely eliminates this redundancy.
Our theoretical result is that MorphisHash saves $\Theta(\ln(n))$ bits in expectation compared to ShockHash.
This corresponds to a factor of 20 less space overhead in practice.
Just like ShockHash, MorphisHash can be used as a building block within RecSplit to obtain MorphisHash-RS.
When compared for same space consumption, MorphisHash-RS can be constructed up to \speedup{} times faster than ShockHash-RS.
The technique to accomplish this might be of a more general interest to compress data structures.

\end{abstract}

\newpage
\section{Introduction}
Given a set $S$ of $n$ keys, a minimal perfect hash function (MPHF) maps each key to a unique position in $[n]:=\{1, \ldots, n\}$.
MPHFs have a wide range of applications including compressed full-text indexes \cite{belazzougui2014alphabet}, computer networks \cite{lu2006perfect}, databases \cite{chang2005perfect}, prefix-search data structures \cite{belazzougui2010fast}, language models \cite{pibiri2017efficient}, bioinformatics \cite{crawford2018practical,pibiri2022sparse}, and Bloom filters \cite{broder2004network}.
Different techniques exist for constructing an MPHF.
They offer a variety of trade-offs between construction time, space consumption and query time.
The space lower bound of an MPHF is $\log_2(e)\approx 1.443$ bits per key \cite{mehlhorn1982program}.

\myparagraph{ShockHash.}
Our technique builds on ShockHash \cite{lehmann2023bipartite, lehmann2023shockhash}.
Similar to Cuckoo hashing \cite{pagh2004cuckoo}, each key is given two candidate positions using respective hash functions $h_{s,0}: S \rightarrow [n]$ and $h_{s,1}: S \rightarrow [n]$ with seed $s$.
ShockHash finds a seed $s$ such that all keys can be mapped to one of their candidate positions and there is exactly one key mapped to each position.
ShockHash needs to store the seed $s$, once found.
Additionally, it needs to store for each key $k\in S$ if the candidate position $h_{s,0}(k)$ or $h_{s,1}(k)$ is used.
This can be represented using a function $f: S \rightarrow \{0, 1\}$.
Such a mapping can be stored efficiently using a retrieval data structure \cite{dillinger2021fast} which requires about 1 bit per key.
A key $k \in S$ is queried using $h_{s,f(k)}(k)$.

A different perspective on ShockHash is that with each seed $s$ it samples a random graph.
The $[n]$ possible output positions are the nodes of that graph.
The keys are the edges, connecting the nodes of their respective candidate positions.
A seed is accepted if the graph can be oriented, i.e. each edge is given a direction, such that the indegree of each node is $1$.
This is possible if and only if the graph is a \emph{pseudoforest} – a graph where each component is a cycle with trees branching from it.
The edges of the cycle of each component are oriented either all `clockwise' or `counterclockwise'.
The edges in the trees are uniquely oriented away from their cycle.
Hence, the indegree of each node is 1.
ShockHash arbitrarily chooses one of two possible orientations of each cycle and stores the according orientation of each edge in a retrieval structure.
Hence, for each cycle there is one bit of redundancy.

\myparagraph{Contribution.}
In this paper we introduce MorphisHash which exhausts this remaining redundancy.
MorphisHash is a recursive acronym: \textbf{M}orphisHash is an \textbf{o}verloaded \textbf{r}etrieval structure for \textbf{p}erfect \textbf{h}ash\textbf{i}ng using \textbf{S}hock\textbf{Hash}.
Our key observation is that the possible orientations of a pseudoforest can be described as the solution of a linear equation system.
A retrieval structure that stores the edge orientations can also be described as the solution of a linear equation system.
This allows us to concatenate the equation systems using matrix multiplication.
We achieve compression by reducing the dimensionality of the solution space of the combined equation system.
Our theoretical insight is that a random pseudoforest has $\Theta(\ln(n))$ components in expectation and MorphisHash can convert this into $\Theta(\ln(n))$ bits of expected space savings compared to ShockHash by utilizing the freedom of choosing the orientation of each component's cycle.
Our experiments show that MorphisHash has about a factor of $20$ less space overhead than ShockHash at the cost of roughly $4$ times more construction time.

\myparagraph{Partitioning.}
Note that within this paper, $n$ denotes the input size for one instance of MorphisHash.
In \cref{s:partitioning}, an MPHF is obtained by splitting a large input key set into a linear number of MorphisHash instances, each of size $n$, and concatenating them afterwards.
Partitioning is required mainly for keeping construction times feasible.
Furthermore, our per instance space savings translate to linear space savings in terms of the large input.

\myparagraph{Outline.}
We begin in \cref{s:rl} with related space efficient PHF construction techniques.
We present MorphisHash in \cref{s:mh} and analyze it in \cref{s:analysis}.
We explain implementation details in \cref{s:partitioning}.
Finally, \cref{s:experiments} discusses experiments and the paper is concluded with \cref{s:conclusion}.
Our compression technique might be of a more general interest and we give further examples \fvref{s:appl}.

\section{Related Work}
\label{s:rl}
We provide a brief overview of other space efficient PHF techniques.
For a detailed survey of state-of-the-art techniques we refer to \cite{lehmann2025modern}.

\myparagraph{RecSplit}
RecSplit \cite{esposito2020recsplit, bez2022high} first hashes the keys into partitions of about 2000 keys.
RecSplit then finds a seed of a hash function that splits the partition into smaller subsets of equal size.
This is applied recursively resulting in a tree-like structure.
Once sufficiently small subsets are obtained, RecSplit uses brute-force search to find an MPHF within each leaf.
Very recently a significant improvement to RecSplit has been made with the introduction of \consensus\cite{lehmann2025combined}.
Instead of allowing arbitrarily large seeds, \consensus-RecSplit uses a fixed number of bits for each seed and backtracks in the splitting tree if a seed cannot be found within the allowed space.
\consensus-RecSplit is currently the most space efficient technique with just 0.001 bits per key overhead.

\myparagraph{PHOBIC}
Another PHF construction technique is PHOBIC \cite{PHOBIC}.
Again, the keys are hashed into partitions of about 2000 keys.
Within each partition the keys are hashed to buckets which have an average size of about 10.
For each bucket, PHOBIC uses brute force search to find a seed of a hash function such that all keys of that bucket are hashed to positions in $[n]$ to which no keys of previous buckets are hashed to.
The buckets are inserted in non-increasing order of size because it is much easier to insert the larger buckets when the output domain is almost empty.
This effect is utilized further by intentionally making some of the buckets larger.
PHOBIC has fast queries at the cost of more space overhead.

\section{MorphisHash}
\label{s:mh}
ShockHash samples random graphs until stumbling on a pseudoforest.
The only remaining degree of freedom when orienting the pseudoforest is that each component contains a cycle and there are two ways to orient each cycle.
We address this remaining redundancy with MorphisHash.

The first ingredient of MorphisHash is the insight that all allowed orientations of a graph can be expressed as an affine subspace of $\mathbb{F}_2^S$, where $S$ is the set of $n$ keys.
To show this, we define $y \in \mathbb{F}_2^S$ as the vector representing the orientation of each edge such that an edge $j \in S$ is oriented to node $h_{s,y_j}(j)$.
We now consider a given pseudoforest and one possible orientation $y$.
We can flip the orientation of a cycle by adding a vector $v \in \mathbb{F}_2^S$ to y where $v_j=1$ if and only if $j$ is part of that cycle.
This can be done for each cycle independently.
Clearly, the dimension of this subspace is equal to the number of components.
Part of this section is to describe the linear equation system of which the solution space is our desired affine subspace.

The second ingredient is a 1-bit retrieval data structure.
The retrieval structure works by storing a bit vector $x \in \mathbb{F}_2^b$, where the parameter $b \in \mathbb{N}_0$ is discussed in detail later.
The orientation of an edge $j\in S$ is described using $h'_s(j) x$, where $h'_s: S \rightarrow \mathbb{F}_2^{1 \times b}$ is a hash function and $s$ is the ShockHash seed.
Using linear equations that involve hash functions is a common technique for retrieval data structures \cite{dillinger2021fast}.

The beauty of MorphisHash is that we can concatenate both linear equation systems using a simple matrix multiplication to find a retrieval structure which directly orients the edges correctly.
We can then decrease the number of bits $b$ that the retrieval structure is allowed to use which reduces the dimension of the solution space and therefore extracts the remaining redundancy.

We now show the equation system that describes the allowed edge orientations.
As a first step we show that the constraints of an MPHF can be weakened in the following sense:

\begin{lemma}
	\label{l:unevenmphf}
	A function $f: S \rightarrow [n]$ with $|S| = n$ is an MPHF (i.e. a bijection) if and only if for all $i \in [n]$ we have that $|\{f^{-1}(i)\}|$ is uneven.
\end{lemma}
\begin{proof}
	Clearly, if $f$ is bijective then $|\{f^{-1}(i)\}| = 1$ is uneven.
	If $f$ is not bijective then there is at least one $i$ such that $|\{f^{-1}(i)\}| = 0$ which is even.
\end{proof}

\myparagraph{Linear Equations in Graphs.}
This allows us to count the indegree of each node using $\mathbb{F}_2$: If the indegree of all nodes is $1_{\mathbb{F}_2}$ then the orientations result in a valid MPHF.
Recall the definition of $y \in \mathbb{F}_2^S$ as the vector representing the orientation of each edge such that an edge $j \in S$ is oriented to node $h_{s,y_j}(j)$.
A different perspective is that an edge $j$ is oriented towards a node $i$ if and only if $(h_{s,1}(j) = i \land y_j = 1) \lor (h_{s,0}(j) = i \land  y_j + 1 = 1)$, which will be useful in the following equation.
We define $A\in \mathbb{F}_2^{n \times S}$ as the incidence matrix of the graph.
We also define $d \in \mathbb{F}_2^n$ where $d_i = |\{j~|~h_{s,0}(j) = i\}| + 1$ counts the number of edges (+1) that are mapped to position $i$ using the $h_{s,0}$ candidate function.
Finally, this allows us to count the indegree of a node $i$ as 
\begin{align*}
	\sum_{h_{s,1}(j) = i} y_j + \sum_{h_{s,0}(j) = i}( 1 + y_j) = \big(\sum_{\substack{h_{s,0}(j) = i \\ \oplus h_{s,1}(j) = i }} y_j \big) + d_i + 1 = A_i y + d_i + 1
\end{align*}

The $\oplus$ operator is logical XOR.
According to \cref{l:unevenmphf} setting $A_i y + d_i + 1 = 1$ such that the indegree of each node $i$ is uneven results in a valid orientation of all edges.
The complete equation system therefore simplifies to $Ay=d$, which has a solution if and only if the graph is a pseudoforest.
Note that in case of a loop, the respective column in the incidence matrix is a zero column.

\begin{figure}[t]
	\centering
	\includegraphics[width=\linewidth]{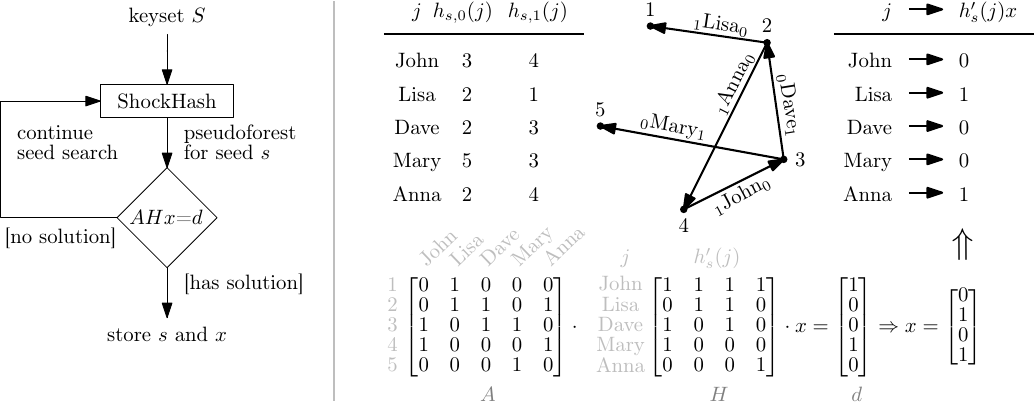}
	\caption{Left: MorphisHash uses ShockHash as a black box.
		Right: An example.
		ShockHash has found a pseudoforest using the candidate positions $h_{s,0}$ and $h_{s,1}$ as shown in the table and also illustrated as a graph.
		The description $_1\text{Lisa}_0$ indicates that $h_{s,1}$ is the left candidate and $h_{s,0}$ the right.
		This can be used to find $d$.
		For example, $d_3 = 0^{\text{Dave}} + 0^{\text{Mary}} + 1^{\text{John}} + 1^{\text{offset}}  = 0$.
		The incidence matrix $A$ and some random matrix $H$ is shown.
		A solution $x$ is found resulting in the orientations as shown both in the right table and using arrows in the graph. 
		The solution $x$ describes the orientation of 5 edges but only requires 4 bits to store.
		The figure is derived from \cite[Fig. 1]{lehmann2023shockhash}.	
}
	\label{fig:mhex}
\end{figure}

\myparagraph{The Retrieval Data Structure.}
To store the orientation of the edges we cannot use $y$ directly because we do not know the index of a key during query time.
We therefore employ the idea of a retrieval structure.
Our retrieval structure consists of a bit vector $x \in \mathbb{F}_2^b$, where $b$ is a tuning parameter.
The orientation of edge $j$ is then described using a scalar product $y_j = h'_s(j) x$ where $h'_s: S \rightarrow \mathbb{F}_2^{1 \times b}$ is a hash function and $s$ is the ShockHash seed.
MorphisHash needs to find $x$ such that all edges are properly oriented.
The above linear equation is given for each key and can therefore be written using a matrix $H$.
Each row $H_j$ is the hash $h'_s(j)$ of the respective key $j$.
We have $y=Hx$ and substitute it into $Ay=d$ resulting in $AHx=d$.
If a solution $x$ for $AHx=d$ exists then the graph is a pseudoforest, $y=Hx$ is a valid orientation and $x$ requires exactly $b$ bits to store using a bit string.
We discuss the selection of $b$ both in theory (\cref{s:analysis}) and practice (\cref{s:experiments}).
Querying our MPHF for a key $j$ is now straightforward $h_{s, h'_s(j) x}(j)$.

If the system $AHx=d$ does not have a solution for a sampled graph this can have two reasons.
(1) The graph is no pseudoforest 
(2) A possible orientation of the pseudoforest is not within the solution space of the retrieval structure.
If there is no solution we reject the seed and ShockHash continues with searching for a new seed.

The already existing step in ShockHash of checking whether a seed is a pseudoforest is therefore redundant.
However, solving an equation system is computationally more expensive than using the original ShockHash pseudoforest check.
We therefore leave the original pseudoforest check as a filter and only need to solve $AHx=d$ a few times in practice.
MorphisHash is illustrated using an example in \cref{fig:mhex}.

\myparagraph{Bipartite MorphisHash.}
A variant of ShockHash is bipartite ShockHash.
Assume that $n$ is even, the extension for uneven numbers can be found in the original paper \cite{lehmann2023bipartite}.
In bipartite ShockHash, the ranges of the two hash functions are made disjoint using $h_{s_0, 0}: S \rightarrow \{1, \ldots, \frac{n}{2}\}$ and $h_{s_1, 1}: S \rightarrow \{\frac{n}{2} + 1, \ldots, n\}$.
The seed is also split in two independent parts $s_0$ and $s_1$.
Seeds where not all positions are hit by at least one candidate position are filtered out.
Bipartite ShockHash only checks if pairs of $s_0$ and $s_1$ that passed the filter result in a pseudoforest.
MorphisHash uses ShockHash as a black box and can be applied on the bipartite case just as well as on the non-bipartite case.
In an obvious manner, we refer to bipartite MorphisHash if bipartite ShockHash is used.

\section{Analysis}
\label{s:analysis}
In this section we analyze non-bipartite MorphisHash for large $n$.
Our analysis is split in two parts.
First, we analyze the number of components of a random pseudoforest.
We use this in the second part, to show the space efficiency of MorphisHash compared to ShockHash.
Note that we employ the common simple uniform hashing assumption \cite{pagh2008uniform}, which assumes that hash functions behave like truly random functions. 

\subsection{The Number of Components in a Random Pseudoforest}
In a pseudoforest each component is a tree with one additional edge.
An equivalent view is that each component of the pseudoforest is a cycle with trees branching from it.

Our first result uses the following graph model.
Let $R=[n] ^ {[n]}$ be the set of all functions from $[n]$ to $[n]$.
Every $r \in R$ corresponds to a directed graph $G_r = ([n], \{(v, r(v)) ~|~ v \in [n]\})$.
All $G_r$ are pseudoforests, because if we start from any node and follow its edges we will eventually end up in a cycle.
Furthermore, there can only be one cycle in each component because each node has an out-degree of one.
In the following we refer to all $G_r$ as \emph{maximal directed pseudoforests}.
All edges of the pseudoforest which are not part of the cycle are pointed towards the cycle and all edges in the cycle are all directed in the same direction.
We will first need the following results.

\begin{lemma}[\cite{takacs1990cayley}]
	\label{l:roottrees}
	Let $T(n, i)$ be the number of undirected forests having node set $[n]$ with $i$ components where nodes $1, 2, \ldots, i$ belong to different trees.
	We have $T(n, i) = in^{n-i-1}$.
\end{lemma}

\begin{lemma}
	\label{l:cycles}
	$|\{ r \in R \ ~|~ \text{all components in}~G_r~ \text{are cycles}\}| = n!$
\end{lemma}
\begin{proof}
	If each node in $G_r$ is part of a cycle then the indegree of each node is 1.
	Hence, $r$ is a bijection and there are $n!$ possible bijections.
\end{proof}
We can now combine the previous results to analyze how the number of nodes in cycles is distributed in maximal directed pseudoforests.
\begin{lemma}
	\label{l:sub-graph-size}
	$|\{ r \in R \ ~|~ G_r~\text{is pseudoforest and has i nodes in cycles} \}| = \frac{i n^{n-1-i} n!}{(n-i)!}$.
\end{lemma}
\begin{proof}
	There are $\binom{n}{i}$ ways of partitioning the nodes into (1) cycles with a total of $i$ nodes and (2) trees that are attached to the cycles.
	According to \cref{l:cycles} there are $i!$ ways to create cycles with $i$ nodes.
	The $i$ nodes in the cycles are the roots of trees.
	According to \cref{l:roottrees} the number of labeled rooted trees with $n$ nodes and $i$ roots is $in^{n-i-1}$ .
	This number holds for graphs where each node has one outgoing edge, because all edges are uniquely oriented towards their parent node.
	We therefore have 
	\begin{align*}
		\binom{n}{i} i! in^{n-i-1} = \frac{n!}{i! (n-i)!} i!  in^{n-i-1} = \frac{i n^{n-1-i} n!}{(n-i)!}
	\end{align*}
	
\end{proof}
The previous result is based on maximal directed pseudoforests.
However, in MorphisHash we sample graphs in a different manner.
Let $Q = [n]^{S \times \{0,1\}}$ be the set of functions from $S \times \{0,1\}$ to $[n]$.
Each $q \in Q$ corresponds to a graph $G_q = ([n], \{\{q(x,0),q(x,1)\} ~|~ x \in S\})$.
We refer to the set of all $G_q$ as the \emph{hashed} graph model.
Graphs in this model may have multiple edges and loops.
MorphisHash uniformly samples functions $q\in Q$.
We are interested in the distribution of $G_q$ conditioned on the event that $G_q$ is a pseudoforest.
In the following, we transfer our result of maximal directed pseudoforests to the hashed graph model.

\begin{lemma}
	\label{l:sample-prob}
	Let $c(G)$ be the number of components of $G$, $C_n$ an appropriate normalization constant, $q \sim \mathcal{U}(Q)$ and $r \sim \mathcal{U}(R)$ then
	\begin{align*}
		p_n(i):=&\mathbb{P}[G_q~\text{has i nodes in cycles}~|~G_q~\text{is pseudoforest}]\\
		= &C_n \mathbb{E}[2^{-c(G_r)}~|~G_r~\text{has i nodes in cycles}] \frac{i}{n^i (n-i)!}
	\end{align*}
\end{lemma}
\begin{proof}
	We begin with a relation between maximal directed pseudoforests $G_r$ and hashed pseudoforests $G_q$.
	A hashed pseudoforest $G_q$ is related to a maximal directed pseudoforest $G_r$ if we can obtain $G_r$ by orienting the edges of $G_q$.
	Let $a_{k,d,l}$ be the number of maximal directed pseudoforests and $b_{k,d,l}$ the number of hashed pseudoforest with $i$ nodes in cycles, $k$ components, $d$ double edges and $l$ loops.
	Within this subset, each cycle of the hashed graph can be oriented in two possible ways to obtain a maximal directed pseudoforest.
	However, changing the orientation of cycles that are a loop or double edge does not result in different maximal directed pseudoforests.
	Hence, we have $\frac{a_{k,d,l}}{b_{k,d,l}} = 2^{k-d-l}$ orientations which result in different maximal directed pseudoforest.
	
	Let $c_{k,d,l}$ be the number of elements $q \in Q$ such that $G_q$ is a pseudoforest with $i$ nodes in cycles, $k$ components, $d$ double edges and $l$ loops.
	We show that $\frac{c_{k,d,l}}{b_{k,d,l}} = n!2^{n-d-l}$.
	We have to consider the number of functions in $Q$ that result in the same hashed pseudoforest.
	The order of the edges can be permuted in $n!$ ways without changing the underlying pseudoforest, except for double edges.
	The number of possible permutations decreases by a factor of two for each double edge ($2^{-d}$).
	Analogously, the nodes within the $n$ edges can be switched without changing the underlying graph in $2^n$ ways, except for loops ($2^{-l}$).
	
	By definition we have $|Q_i| = \sum_{k,d,l} c_{k,d,l}$ and $|R_i| = \sum_{k,d,l} a_{k,d,l}$, where $Q_i \subset Q$ and $R_i \subset R$ are the pseudoforests with $i$ nodes in cycles.
	For brevity let \[p_{k,d,l} := \frac{a_{k,d,l}}{|R_i|} = \mathbb{P}[G_r~\text{has $k$ components, $d$ double edges and $l$ loops}~|~G_r~\text{has $i$ nodes cycles}]\]
	We have
	\begin{align*}
		p_n(i)&= \frac{|Q_i|}{\sum_j |Q_j|} =\frac{|R_i|}{\sum_j |Q_j|} \frac{\sum_{k,d,l} c_{k,d,l}}{|R_i|}
		=\frac{|R_i|}{\sum_j |Q_j|}\sum_{k,d,l} \frac{a_{k,d,l} c_{k,d,l}}{a_{k,d,l} |R_i|}\\
		&=\frac{|R_i|}{\sum_j |Q_j|}\sum_{k,d,l} p_{k,d,l} \frac{c_{k,d,l}}{a_{k,d,l}}
		=\frac{|R_i|}{\sum_j |Q_j|}\sum_{k,d,l} p_{k,d,l} \frac{c_{k,d,l}}{b_{k,d,l}}\frac{b_{k,d,l}}{a_{k,d,l}}\\
		&=\frac{|R_i|}{\sum_j |Q_j|}\sum_{k,d,l} p_{k,d,l}  n!2^{n-d-l}2^{-k+d+l}
		=\frac{|R_i|n!2^n}{\sum_j |Q_j|}\sum_{k,d,l} p_{k,d,l}  2^{-k}\\
		&=\frac{|R_i|n!2^n}{\sum_j |Q_j|}\mathbb{E}[2^{-c(G_r)}~|~G_r~\text{has $i$ nodes in cycles}]\\
		&=C_n \mathbb{E}[2^{-c(G_r)}~|~G_r~\text{has $i$ nodes in cycles}] \frac{i}{n^i (n-i)!}
	\end{align*}
	Where we used \cref{l:sub-graph-size} for $|R_i|$. $\sum_j |Q_j|$ is the number of all hashed pseudoforests and thus only depends on $n$. $C_n$ is chosen such that the probabilities add up to 1.
\end{proof}
The next step is to analyze $\mathbb{E}[2^{-c(G_r)}~|~G_r~\text{has $i$ nodes in cycles}]$.
To this end, we first require some definitions and more general results.

\begin{lemma}
	\label{l:inde}
The number of components of a pseudoforest with $n$ nodes and $i$ nodes in cycles follows the same distribution as the number of components of a graph of $i$ nodes where each component is a cycle.
\end{lemma}
Note that this refers to both graph models $G_r$ and $G_q$.
\begin{proof}
	The order in which the edges of a graph are sampled does not change the distribution of the graph.
	Given a graph of $n$ nodes and $i$ edges such that all edges are part of cycles,
	the probability that the remaining $n-i$ edges form trees with the root being part of the existing cycles is independent of how the nodes of the cycles are connected to form any number of components.
\end{proof}

\myparagraph{Configuration Model.}
The \emph{configuration model} \cite{newman2010networks} can be used to describe distributions of random graphs.
In the model, each node is given a fixed number of half-edges.
The graph is obtained by repeatedly connecting half-edges by uniformly sampling from the set of all remaining half-edges.

\begin{lemma}
	\label{l:rounds1}
	The distribution of the hashed graph $G_q$ with $q \sim \mathcal{U}(Q)$ is equal to the distribution of the following graph.
	
	The graph is obtained in two steps.
	First, the degree of each node is revealed by distributing $2n$ half-edges.
	In a second step, the edges of the graph are obtained in a sequence of $n$ rounds. In each round an unmatched half-edge $i \in [2n]$ is chosen arbitrarily and matched to a distinct unmatched half-edge $j \in [2n]$, chosen uniformly at random. The choice of $i$ may depend on the set of half-edges matched previously.
\end{lemma}
We refer to ShockHash \cite[Lemma 5]{lehmann2023bipartite} for a proof.
We require an analogous result for $G_r$
\begin{lemma}
	\label{l:rounds2}
	The distribution of the maximal directed pseudoforest $G_r$ with $r \sim \mathcal{U}(R)$ is equal to the distribution of the following directed graph.
	
	The number of outgoing edges of each node is fixed to 1.
	The graph is obtained in two steps.
	First, the indegree of each node is revealed by distributing $n$ incoming half-edges.
	In a second step, the edges of the graph are obtained in a sequence of $n$ rounds. In each round an unmatched outgoing half-edge at node $i \in [n]$ is chosen arbitrarily and matched to a distinct unmatched incoming half-edge $j \in [n]$, chosen uniformly at random. The choice of $i$ may depend on the set of half-edges matched previously.
\end{lemma}
Again, the proof is analogous to ShockHash \cite[Lemma 5]{lehmann2023bipartite}.

\begin{lemma}
	\label{l:prod}
	For $i\in \mathbb{N}: \prod_{k=2}^{i} (1 - \frac{1}{k})= \frac{1}{i}$
\end{lemma}
\begin{proof}
	By induction $\prod_{k=2}^{1} (1 - \frac{1}{k}) = 1$ and
	\begin{align*}
	\prod_{k=2}^{i + 1}(1 - \frac{1}{k}) = (1 - \frac{1}{i + 1}) \prod_{k=2}^{i} (1 - \frac{1}{k})= (1 - \frac{1}{i + 1}) \frac{1}{i} = \frac{1}{i + 1}
\end{align*}
\end{proof}
We now have the available tools to analyze the number of components conditioned on the number of nodes in cycles.
\begin{lemma}
	\label{l:bounds}
	$\frac{1}{4}i^{-1/2} \leq \mathbb{E}[2^{-c(G_r)}~|~G_r~\text{has $i$ nodes in cycles}] \leq i^{-1/2}$,
	where $r\sim \mathcal{U}(R)$ and $c(G_r)$ is the number of components of $G_r$.
\end{lemma}
Note that the following proof has similarities to ShockHash \cite[Lemma 6]{lehmann2023bipartite}.
\begin{proof}
	We consider the number of components of a maximal directed pseudoforest $G_r^i$ with $i$ nodes, $R=[i]^{[i]}$, $r\sim \mathcal{U}(R)$, conditioned on the event that each component of $G_r^i$ is a cycle.
	According to \cref{l:inde} the number of components of $G_r^i$ follows the same distribution as the number of components of a maximal directed pseudoforest with $n$ nodes and $i$ nodes in cycles.
	We proceed as described in \cref{l:rounds2} by revealing the graph in a sequence of rounds.
	First, the indegree of each node is revealed.
	Each component is a circle and therefore each node has indegree one.
	We choose the outgoing half-edge at an arbitrary node $x$.
	The outgoing half-edge is matched with one of the $i$ incoming half-edges.
	Let $y$ be the node at the incoming half-edge.
	There are two cases.
	\begin{enumerate}
		\item With probability $\frac{1}{i}$ we have $x=y$.
		In this case, a loop is created, a cycle is closed and the next outgoing half-edge to match is chosen arbitrarily.
		\item Otherwise we merge the nodes to a single one which does not change the number of components. The outgoing half-edge at node $y$ is matched next.
	\end{enumerate}
	In both cases, the distribution of the remaining graph is that of $G_r^ {i-1}$.
	Because of this independence, we can multiply the expectation values to obtain the recurrence
	\begin{align*}
		\mathbb{E}[2^{-c(G_r^i)}] &=\frac{1}{2} \frac{1}{i} \mathbb{E}[2^{-c(G_r^{i-1})}]  + \left(1-\frac{1}{i}\right) \mathbb{E}[2^{-c(G_r^{i-1})}]= \left(1- \frac{1}{2i}\right) \mathbb{E}[2^{-c(G_r^{i-1})}]
	\end{align*}
	With the base case $\mathbb{E}[2^{-c_0}] = 1$ the recurrence is solved and upper bounded by	
	\begin{align*}
		\mathbb{E}[2^{-c(G_r^i)}]&=\prod_{k=1}^{i} 1-\frac{1}{2k} = \sqrt{\prod_{k=1}^{i} \left(1-\frac{1}{2k}\right)^2}\leq \sqrt{\prod_{k=1}^{i} \left(1-\frac{1}{2k}\right) \left(1-\frac{1}{2k + 1}\right)}\\
		&=\sqrt{\prod_{k=2}^{2i + 1} \left(1-\frac{1}{k}\right)} = \sqrt{\frac{1}{2i+1}} \leq i^{-1/2}
	\end{align*}
	Where we used \cref{l:prod}. Analogously, a lower bound is
	\begin{align*}
		\mathbb{E}[2^{-c(G_r^i)}]& = \sqrt{\left(1-\frac{1}{2}\right)^2\prod_{k=2}^{i} \left(1-\frac{1}{2k}\right)^2}= \frac{1}{2}\sqrt{\prod_{k=2}^{i} \left(1-\frac{1}{2k}\right)^2}\\
		&\geq \frac{1}{2}\sqrt{\prod_{k=2}^{i}\left(1-\frac{1}{2k-1}\right) \left(1-\frac{1}{2k-2}\right)} \geq \frac{1}{2}\sqrt{\prod_{k=2}^{2i}\left(1-\frac{1}{k}\right)}\\
		&=\frac{1}{2}\sqrt{\frac{1}{2i}} \geq \frac{1}{4} i^{-1/2}
	\end{align*}
\end{proof}
A similar result for hashed pseudoforests instead of maximal directed pseudoforests is the following.
\begin{lemma}
	\label{l:cc}
	$\mathbb{P}[c(G_q)\geq\frac{1}{4}\ln(i)~|~G_q~\text{is pseudoforest with $i$ nodes in cycles}] \in \Omega_{i \rightarrow \infty}(1)$, where $q\sim \mathcal{U}(Q)$ and $c(G_q)$ is the number of components of $G_q$.
\end{lemma}
\begin{proof}
	The proof is analogous to the previous one.
	We use \cref{l:rounds1} to match a half-edge to one of the remaining $2i-1$ half-edges.
	In each round the number of components increases by one with probability $\frac{1}{2i-1}$.
	Resolving the respective recurrence results in
	\begin{align*}
		&\mathbb{E}[c(G_q)~|~G_q~\text{is pseudoforest with $i$ nodes in cycles}]\\
		=&\phantom{\frac{1}{2}} (\frac{1}{2i-1} + \frac{1}{2i-3} + \frac{1}{2i-5} + \frac{1}{2i-7} + \ldots + \frac{1}{3} + 1) \\			    
		=&\frac{1}{2} (\frac{1}{2i-1} + \frac{1}{2i-1} + \frac{1}{2i-3} + \frac{1}{2i-3} + \ldots + 1 + 1) \\
		>&\frac{1}{2} (\frac{1}{2i-0} + \frac{1}{2i-1} + \frac{1}{2i-2} + \frac{1}{2i-3} + \ldots + \frac{1}{2} + 1)\\
		=&\frac{1}{2} H_{2i} > \frac{1}{2}\ln(i)
	\end{align*}
	Where $H_k$ is the k-th harmonic number.
	We are interested in a lower bound for the probability that the number of components is at least $\frac{1}{4}\ln(i)$.
	We can use the variance to find such a bound.
	The probability distribution of the number of components can be described in terms of a Poisson binomial distribution.
	The variance of a Poisson binomial distribution can be bounded above by the expected value.
	Clearly, $H_{2r} < \ln(3i)$ is an upper bound for the variance.
	Using Cantelli's inequality ($\mathbb{P}[X\geq \mathbb{E}[X] + x] \geq \frac{x^2}{x^2+\sigma^2}$) we find with $x=-\frac{1}{4}\ln(i)$ that with probability $\frac{(\frac{1}{4}\ln(i))^2}{(\frac{1}{4}\ln(i))^2 + \ln(3i)} \in \Omega(1)$ the pseudoforest has at least $\frac{1}{4}\ln(i)$ components.
\end{proof}
The previous result shows that the number of components increases logarithmically in the number of nodes in cycles.
The next step is therefore to show that the cycles are sufficiently large.
\begin{lemma}
	\label{l:cyc}
	$\mathbb{P}[G_q~\text{has at least $\sqrt{n}$ nodes in cycles}~|~G_q~\text{is pseudoforest}] \in \Omega_{n\rightarrow \infty}(1)$, where $q\sim \mathcal{U}(Q)$.
\end{lemma}
\begin{proof}
	Let $p_n(i)$ be the probability of sampling a hashed pseudoforest with $i$ nodes in cycles as determined in \cref{l:sample-prob}. We need to show that the probability of sampling a pseudoforest with at least $\sqrt{n}$ nodes in cycles is at least a constant factor larger than sampling a pseudoforest with less than $\sqrt{n}$ nodes in cycles, i.e. $\left(\sum_{i=\sqrt{n}}^{n} p_n(i)\right) / \left(\sum_{i=1}^{\sqrt{n}} p_n(i)\right)\in\Omega(1)$.\\
	A stronger statement is ${\left(\sum_{i=\sqrt{n}}^{2\sqrt{n}} p_n(i)\right) / \left(\sum_{i=1}^{\sqrt{n}} p_n(i)\right) \in \Omega(1)}$.
	In the following we show $\frac{p_n(\sqrt{n}+x)}{p_n(\sqrt{n}-x)} \in \Omega(1)$ for all $0 \leq x \leq \sqrt{n}$.
	This pointwise comparison is an even stronger statement.
	For brevity, we omit rounding of $\sqrt{n}$.
	Using the bounds of \cref{l:bounds} and Stirling's approximation we have:
	\begin{align*}
		&\frac{p_n(\sqrt{n}+x)}{p_n(\sqrt{n}-x)} \geq \frac{1}{4} \frac{(\sqrt{n}-x)^{1/2}}{(\sqrt{n}+x)^{1/2}} \cdot \frac{\sqrt{n}+x}{\sqrt{n}-x}\cdot \frac{n^{\sqrt{n}-x}}{n^{\sqrt{n}+x}}\cdot \frac{(n-\sqrt{n}+x)!}{(n-\sqrt{n}-x)!}\\
		&=\frac{1}{4}\left(\frac{\sqrt{n} + x}{\sqrt{n} - x}\right)^{1/2} \cdot n^{-2x} \cdot \sqrt{\frac{n-\sqrt{n} + x}{n-\sqrt{n} - x}} \cdot \frac{e^{n-\sqrt{n} - x}}{e^{n-\sqrt{n} + x}} \cdot \frac{(n-\sqrt{n} + x)^{n-\sqrt{n} + x}}{(n-\sqrt{n} - x)^{n-\sqrt{n} - x}}\\
		&> \frac{1}{4} n^{-2x} \cdot e^{-2x} \cdot \left(\frac{n-\sqrt{n} + x}{n-\sqrt{n} - x}\right)^{n-\sqrt{n}} \cdot \left((n-\sqrt{n} + x)(n-\sqrt{n} - x)\right)^x\\
		&= \frac{1}{4} e^{-2x} \cdot \left(\frac{1+\frac{x}{n-\sqrt{n}}}{1-\frac{x}{n-\sqrt{n}}}\right)^{n-\sqrt{n}} \cdot \left(\frac{(n-\sqrt{n} + x)(n-\sqrt{n} - x)}{n^2}\right)^x\\
		&>\frac{1}{4} e^{-2x} \cdot \frac{e^x}{e^{-x}} \cdot \left(\frac{(n-\sqrt{n} + \sqrt{n})(n-\sqrt{n} - \sqrt{n})}{n^2}\right)^{\sqrt{n}}\\
		&=\frac{1}{4} \left(\frac{n(n-2\sqrt{n})}{n^2}\right)^{\sqrt{n}} = \frac{1}{4} \left(1-\frac{2}{\sqrt{n}}\right)^{\sqrt{n}} = \frac{1}{4} \frac{1}{e^2} 
	\end{align*}
\end{proof}
Combining the results gives us the following.
\begin{theorem}
	\label{t:pc}
	$\mathbb{P}[c(G_q)\geq \frac{1}{8}\ln(n)~|~ G_q~\text{is pseudoforest}] \in \Omega_{n\rightarrow \infty}(1)$, where $c(G_q)$ is the number of components of $G_q$ and $q \sim \mathcal{U}(Q)$.
\end{theorem}
\begin{proof}
	According to \cref{l:cyc} at least $\sqrt{n}$ nodes are in cycles with constant probability.
	Using \cref{l:cc}  this results in $\frac{1}{4}\ln(\sqrt{n}) = \frac{1}{8}\ln(n)$ components with constant probability.
\end{proof}

\begin{corollary}
	$\mathbb{E}[c(G_q)~|~ G_q~\text{is pseudoforest}] \in \Theta_{n\rightarrow \infty}(\ln(n))$, where $c(G_q)$ is the number of components of $G_q$ and $q \sim \mathcal{U}(Q)$.
\end{corollary}
\begin{proof}
	Use \cref{t:pc} for the lower bound.
	There can be at most all $n$ nodes in cycles of the pseudoforest and we can apply \cref{l:cc} for the upper bound.
\end{proof}

\subsection{Space Savings of MorphisHash}
We now show that MorphisHash can convert each component into space savings compared to ShockHash.
As a first step we transition the previous graph results into the world of matrices.

\begin{lemma}
	\label{l:pm}
	The defect of the incidence matrix $A$ of a pseudoforest is at least the number of the pseudoforests components.
\end{lemma}

\begin{proof}
	For each component $C$, summing up the rows $A_j$ of the respective nodes $j \in C$ results in zero rows because for each component, the two endpoints of each edge are included in the summed up rows.
	The rows summed up for each zero row are disjoint and therefore in particular linearly independent combinations.
\end{proof}
\begin{lemma}[\cite{bourgain2010singularity}]
	\label{l:msqp}
	The probability that a random square matrix (of any field) with $n$ rows and columns has full rank approaches 1 for large $n$.
\end{lemma}
\begin{lemma}
	\label{l:mrqp}
	The probability that a random rectangular matrix (of any field) with $n$ rows and $m \in o(n)$ columns has full rank approaches $1$ for large $n$.
\end{lemma}
\begin{proof}
	A necessary condition that a square matrix with $n$ rows has full rank is that the first $m$ columns have full rank.
	The probability that this rectangular submatrix has full rank is therefore bounded below by \cref{l:msqp}.
\end{proof}
We now show our main result.
\begin{theorem}
	\label{l:sp}
	$\mathbb{P}[\exists x: AHx=d~|~G_q~\text{is pseudoforest}]\in \Omega_{n\rightarrow \infty}(1)$, where $q \sim \mathcal{U}(Q)$, incidence matrix $A$ and vector $d$ of $G_q$ are as described in the algorithm, $b=n-\frac{1}{9}\ln(n)$ and $H \sim \mathcal{U}(\mathbb{F}_2^{n \times b})$.
\end{theorem}
\begin{proof}
	In \cref{t:pc} we showed that a hashed pseudoforest has at least $\frac{1}{8}\ln(n)$ components with constant probability.
	According to \cref{l:pm} a direct consequence is that the incidence matrix $A$ of a random pseudoforest has a defect of at least $\frac{1}{8}\ln(n)$ with constant probability.
	According to \cref{l:mrqp}, the probability that $H$ has full rank is at least constant.
	The system $AHx=d$ has a solution if there is a vector $y$ which solves $Ay=d$ and simultaneously $Hx=y$.
	Such a vector $y$ exists if the two solution spaces have to intersect, which happens once $\defe(A)+\rank(H) - n > 0$.
	Since $A$ and $H$ are uncorrelated we have with constant probability that both $H$ has full rank $b$ and simultaneously $A$ has at least a defect of $\frac{1}{8}\ln(n)$. With $b=n-\frac{1}{9}\ln(n)$ we have $\defe(A) +\rank(H) - n \geq \frac{1}{8}\ln(n) + (n-\frac{1}{9}\ln(n)) - n  > 0$ with constant probability.
\end{proof}
We use the above result to measure the space savings compared to ShockHash.
\begin{corollary}
	Compared to ShockHash, MorphisHash is at least $\frac{1}{9}\ln(n) - \mathcal{O}(1)$ bits more space efficient in expectation while requiring a constant factor more time.
\end{corollary}
\begin{proof}
	ShockHash requires at least $n$ bits to store the orientation of all keys in a retrieval structure.
	MorphisHash has to store the solution vector $x$ instead, requiring exactly $b$ bits.
	According to \cref{l:sp} we can choose $b=n-\frac{1}{9}\ln(n)$ to obtain the desired space savings.
	However, there is a constant probability that a seed pair has to be rejected because there is no solution for $x$.
	MorphisHash therefore has to check a constant factor more seeds in expectation, consequently increasing the expected space required to store the seed by a constant number of bits.
	Analogously, the construction time grows by a constant factor in expectation.
	The time required to solve an expected constant number of equation systems is dominated by the seed search.
\end{proof}

\section{Partitioning}
\label{s:partitioning}
The time required to find a seed in MorphisHash and ShockHash grows exponentially with $n$.
To keep construction feasible for a large number of keys, we first partition the input into equally sized subsets of manageable size.
For consistency, we use $n$ for the size of those subsets, and refer to them as \emph{base cases} in the following.
MorphisHash is then applied on those base cases.
Note that the following two partitioning schemes are also applied on ShockHash and we refer to their paper for more details \cite[Section 7]{lehmann2023bipartite}.

\myparagraph{MorphisHash-RS.}
We use RecSplit \cite{esposito2020recsplit} to recursively split the input into smaller subsets.
Once sufficiently small subsets are obtained we apply MorphisHash on those subsets as a base case.
RecSplit is space efficient but has significant query time overheads caused by traversing the tree.
In MorphisHash-RS we store the solution vector $x$ of the retrieval structure directly next to the corresponding seed to improve locality.

\myparagraph{MorphisHash-Flat.}
The input keys are first hashed into buckets.
Using thresholds, some keys are bumped such that the bucket does not exceed the desired base case size $n$.
The bumped keys are then used to fill up buckets which did not reach the desired size.
We apply MorphisHash on each bucket.

\section{Experiments}
\label{s:experiments}
In this section we experimentally evaluate MorphisHash.
We show the effect of the new parameter $b$ introduced by MorphisHash.
We then compare MorphisHash-Flat and MorphisHash-RS with state-of-the-art competitors.
Our source code is public under the General Public License \cite{sourceCodeMorphisHash}.
We integrate MorphisHash into an existing benchmark framework \cite{sourceCodeMphfExperiments}, which was used for the comparison with competitors.
The benchmark framework is described in detail in \cite{lehmann2025modern}.
We perform all experiments on a Core i7-11700 CPU which has 48 KiB L1 and 512 KiB L2 data cache per core.
The CPU has a total of 16 MiB L3 cache.
The machine has 64 GiB of dual-channel DDR4-3200 RAM.
Note that our experiments are at a scale where variances are relatively small, we therefore omit them for better readability.

\begin{figure}[t]
	\begin{subfigure}[t]{\textwidth}
		\centering
		\begin{tikzpicture}
			\centering
			\ref*{legendHashEvals}
		\end{tikzpicture}
	\end{subfigure}
	\\
	\begin{subfigure}[b]{0.48\textwidth}
		    \begin{tikzpicture}
        \begin{axis}[
            title={},
            xlabel={$n$},
            ymin=-0.1,
            width=5cm,	
            height=5cm,
            every axis plot post/.append style={mark repeat=2},
            ylabel={Bits space overhead}, %
            cycle list name=myColorList,
          ]

          \addplot coordinates { (10,0.000623) (12,0.001351) (14,0.01701) };
          
          \addplot coordinates { (10,1.98017) (12,2.15936) (14,2.25384) (16,2.31045) (18,2.33948) (20,2.34407) (22,2.34104) (24,2.32289) (26,2.29606) (28,2.26671) (30,2.24846) (32,2.21135) (34,2.18641) (36,2.17118) (38,2.12968) (40,2.11947) (42,2.09761) (44,2.11943) (46,2.02232) (48,2.07459) (50,2.02907) (52,2.02207) (54,2.01841) (56,2.00451) (58,1.96151) (60,1.98313) (62,1.97916) (64,1.96054) (66,1.99383) (68,1.99487) (70,1.98906) };

          \addplot coordinates { (10,0.690483) (12,0.757064) (14,0.797149) (16,0.816655) (18,0.82153) (20,0.816292) (22,0.804155) (24,0.788188) (26,0.768846) (28,0.749424) (30,0.727341) (32,0.707539) (34,0.685614) (36,0.667494) (38,0.647917) (40,0.634827) (42,0.61862) (44,0.602825) (46,0.589153) (48,0.580282) (50,0.568118) (52,0.561065) (54,0.546187) (56,0.543304) (58,0.53279) (60,0.538096) (62,0.545733) (64,0.532941) (66,0.510035) (68,0.534737) (70,0.527264) };

          \addplot coordinates { (10,0.308536) (12,0.267136) (14,0.263788) (16,0.263618) (18,0.260781) (20,0.25499) (22,0.246482) (24,0.236776) (26,0.225516) (28,0.214755) (30,0.203793) (32,0.193806) (34,0.181612) (36,0.172795) (38,0.163637) (40,0.158496) (42,0.147371) (44,0.138758) (46,0.133825) (48,0.125076) (50,0.118109) (52,0.118201) (54,0.106814) (56,0.112332) (58,0.109616) (60,0.096978) (62,0.103668) (64,0.107749) (66,0.099214) (68,0.086879) (70,0.101065) };

        \end{axis}
    \end{tikzpicture}
		\label{fig:internalfa}
		\caption{Idealized space overhead over the lower bound $\log_2(n^n/n!)$. For an average successful seed $s$ we charge $\log_2(s)$ bits plus the bits required for retrieval i.e. $n$ for ShockHash and $b$ for MorphisHash.}
	\end{subfigure}
	\hfill
	\begin{subfigure}[b]{0.48\textwidth}
		    \begin{tikzpicture}
        \begin{axis}[
            xlabel={$n$},
            ylabel={Avg. successful seed},
            ymode=log,
            ymax=3e5,
            width=5cm,	
            height=5cm,
            legend columns=5,
            transpose legend,
            every axis plot post/.append style={mark repeat=2},
            legend to name=legendHashEvals,
          ]

          \addplot coordinates { (10,2756.92) (12,18631.4) (14,128975) };

          \addlegendentry{Brute-force \cite{esposito2020recsplit}};

          \addplot coordinates { (10,3.99194) (12,5.44174) (14,7.31854) (16,9.7896) (18,13.0655) (20,17.3491) (22,23.0548) (24,30.5118) (26,40.4421) (28,53.5326) (30,71.0984) (32,94.2499) (34,125.426) (36,167.287) (38,221.822) (40,296.963) (42,397.646) (44,536.431) (46,698.217) (48,959.886) (50,1270.14) (52,1711.04) (54,2305.2) (56,3085.47) (58,4110.06) (60,5571.01) (62,7549.71) (64,10100.8) (66,13812.4) (68,18648.0) (70,25074.1) };
          
          \addlegendentry{Bip. ShockHash \cite{lehmann2023bipartite}};
          
          \addplot coordinates { (10,7.264445) (12,9.604127) (14,12.7068) (16,16.8066) (18,22.225) (20,29.3963) (22,38.8999) (24,51.5224) (26,68.2844) (28,90.6423) (30,120.335) (32,160.09) (34,212.973) (36,284.004) (38,378.748) (40,506.459) (42,676.996) (44,905.493) (46,1212.13) (48,1626.32) (50,2181.49) (52,2930.78) (54,3927.58) (56,5289.82) (58,7109.51) (60,9594.27) (62,12984.9) (64,17412) (66,23326.6) (68,31741.8) (70,42712.7) };
        
          \addlegendentry{Bip. MorphisHash b=n-3};
          
          \addplot coordinates { (10,18.145) (12,23.2032) (14,30.2991) (16,39.8329) (18,52.5386) (20,69.4504) (22,91.9586) (24,121.987) (26,161.993) (28,215.488) (30,287.001) (32,382.77) (34,510.671) (36,682.435) (38,912.906) (40,1223.37) (42,1636.98) (44,2193.05) (46,2944.93) (48,3948.53) (50,5303.93) (52,7139.05) (54,9570.33) (56,12926.1) (58,17407.9) (60,23363) (62,31550.5) (64,42643.7) (66,57330.1) (68,76988.1) (70,104529) };
          
          \addlegendentry{Bip. MorphisHash b=n-6};

        \end{axis}
    \end{tikzpicture}
		\caption{Average successful seed. Note that we use Bip. ShockHash and Bip. MorphisHash without the quad split technique (see \cite[Section 8.3]{lehmann2023bipartite}).\\}
	\end{subfigure}
	\caption{Space and avg. successful seed for ShockHash, MorphisHash and brute force search.}
	\label{fig:internalf}
\end{figure}

\begin{figure}[t]
	\centering
	\begin{minipage}{0.45\textwidth}
		    \begin{tikzpicture}
        \begin{axis}[
           title={},
           ylabel={Avg. successful seed},
           legend style={at={(0.2,0.6)}, anchor=west},
           ymin=0,
            width=5cm,	
            height=5cm,
            xlabel={Bits space overhead}, %
            cycle list name=myColorList,
          ]
       		\node at (0.5,5300) {$b=n-6$};
       		\node at (2.1,1700) {$b=n$};
          \addplot coordinates { (2.02907,1270.14) };
          \addlegendentry{Bip. ShockHash \cite{lehmann2023bipartite}};
          \addplot coordinates { (0.0865903,5246.71) (0.149891,3796.26) (0.319193,2848.9) (0.567712,2187.31) (0.952822,1759.62) (1.50777,1505.14) (2.27828,1385.98) };
          \addlegendentry {Bip. MorphisHash};
        \end{axis}
    \end{tikzpicture}
		\caption{Comparing ShockHash with the MorphisHash trade-off between time and space for fixed $n=50$ and $b=\{n-6, \ldots, n\}$.}
		\label{fig:diffka}
	\end{minipage}
	\hfill
	\begin{minipage}{0.45\textwidth}
		    \begin{tikzpicture}
        \begin{axis}[
           title={},
           xlabel={n},
           ymin=0,
           xmin =1,
            width=5cm,	
            height=5cm,
            ylabel={Avg. number of components},
            cycle list name=compCycle,
          ]
          \addplot coordinates { (2,1.0) (4,1.06083) (6,1.12148) (8,1.17318) (10,1.21434) (12,1.2489) (14,1.28102) (16,1.30606) (18,1.33333) (20,1.35457) (22,1.37494) (24,1.39584) (26,1.41096) (28,1.43073) (30,1.44234) (32,1.45929) (34,1.47382) (36,1.48513) (38,1.49669) (40,1.50839) (42,1.52047) (44,1.53156) (46,1.54298) (48,1.5514) (50,1.56089) (52,1.57279) (54,1.5792) (56,1.58541) (58,1.59323) (60,1.60375) (62,1.61212) (64,1.61862) (66,1.62407) (68,1.63427) (70,1.63722) };
        \end{axis}
    \end{tikzpicture}
		\caption{Experimentally measured average number of components when sampling bipartite pseudoforests.}
		\label{fig:components}
	\end{minipage}
\end{figure}

\subsection{MorphisHash vs ShockHash without partitioning}
\label{ss:internal}

In the following we compare Bipartite MorphisHash with Bip. ShockHash.
MorphisHash has the additional parameter $b$, which determines the size of the retrieval structure.
Smaller $b$ values result in less space overhead at the cost of more seed tests.
Note that in theory we choose $b=n-\frac{1}{9}\ln(n)$ but in practice where $n$ is small it is more intuitive to work with $n$ minus a small constant.
\cref{fig:internalf} shows that Bip. MorphisHash has to check roughly a constant factor more seeds than Bip. ShockHash when $b$ is fixed to $n$ minus a constant.
For $b=n-6$ this factor is $\approx 4$.
At the same time, MorphisHash almost completely eliminates the remaining space overhead.
MorphisHash comes below 0.1 bits of space overhead when using $b=n-6$ while ShockHash has about 2 bits of space overhead.
Thus, MorphisHash has roughly 20 times less space overhead.
For $n=54$, MorphisHash has $\frac{0.11}{54} \approx 0.002$ bits per key space overhead.

\cref{fig:diffka} gives another perspective.
In this plot, we fix $n$ and vary $b$.
Interestingly, ShockHash even outperforms MorphisHash for $b=n$ .
This is because MorphisHash requires $n$ bits for the retrieval structure in this case just like ShockHash.
However, there is still the chance that the equation system of MorphisHash has no solution resulting in more retries and therefore in a higher space consumption of the seed and a higher construction time.
For smaller $b$ the space overhead approaches 0.
In the extreme case of $b=0$ and respective zero dimensional retrieval vector, MorphisHash is equivalent to simple brute force search because all keys can only use the $h_0$ candidate function.
The average successful seed grows rapidly with smaller $b$ as it is increasingly less likely to stumble upon a pseudoforest which has at least $n-b$ many components.
A pseudoforest with less than $n-b$ components may still result in a solvable (overdetermined) equation system, but this becomes exponentially less likely with every component below $n-b$.
\cref{fig:diffka} uses $n=50$.
The expected number of components for $n=50$ is 1.56 as shown in \cref{fig:components}.

\subsection{Choosing $b$ in MorphisHash-RS and MorphisHash-Flat}
\label{ss:optibapp}
Space improvements in MorphisHash-RS and MorphisHash-Flat can be made either by (1) increasing the base case size $n$ which reduces the space overhead of the partitioning technique or (2) decreasing $b$ which reduces the space overhead of MorphisHash.
In both cases the construction time increases.
We experimented with different values of $b$ to identify the configurations which dominate the construction time and space trade-off.
We determined that $b=n-4$ is a good choice for MorphisHash-RS and $b=n-2$ for MorphisHash-Flat.
MorphisHash-Flat is a less space efficient partitioning technique compared to MorphisHash-RS.
Space savings can be made more easily by increasing $n$ instead of decreasing $b$.

\subsection{Comparison to Competitors}
\label{ss:comp}

\begin{figure}[t]
	\centering
	\begin{tikzpicture}
		\centering
		\ref*{legendEvalParetoQuery}
	\end{tikzpicture}
\vspace{3mm}
	    \begin{tikzpicture}
        \begin{axis}[
            plotEvalPareto,
            xlabel={Overhead (Bits/Key)},
            ylabel={Throughput (Keys/s)},
            xmode=log,
            ymode=log,
            xmin=0.01,
            xmax=1,
            ymin=1.01e4,
            log x ticks with fixed point,
            patch,
            patch type=rectangle,
            shader=flat,
            legend to name=legend2,
            legend style={nodes={scale=0.9, transform shape}},
            legend columns=3,
          ]
          
          \addplot[mark=consensus,color=colorConsensus,solid,mark=none,line width=0] coordinates { (0.01505,502748) (10,502748) (10,19990) (0.01505,19990) (0.011,325188) (10,325188) (10,19990) (0.011,19990) (0.011,170538) (10,170538) (10,19990) (0.011,19990) (0.011,59296) (10,59296) (10,19990) (0.011,19990) (0.011,29714.9) (10,29714.9) (10,19990) (0.011,19990) (0.08762,2.62343e+06) (10,2.62343e+06) (10,19990) (0.08762,19990) (0.0508,1.54756e+06) (10,1.54756e+06) (10,19990) (0.0508,19990) };
          \addplot[mark=square,color=colorBipartiteShockHash,solid,solid,mark=none,line width=0] coordinates { (0.08718,178979) (10,178979) (10,19990) (0.08718,19990) (0.07861,127014) (10,127014) (10,19990) (0.07861,19990) (0.06684,94599.7) (10,94599.7) (10,19990) (0.06684,19990) (0.06039,78145) (10,78145) (10,19990) (0.06039,19990) (0.07216,95551) (10,95551) (10,19990) (0.07216,19990) (0.05751,58515.4) (10,58515.4) (10,19990) (0.05751,19990) };
          \addplot[mark=consensus,color=colorConsensus,solid,mark=none,line width=0] coordinates { (0.13553,3.72842e+06) (10,3.72842e+06) (10,19990) (0.13553,19990) };
          \addplot[mark=square,color=colorBipartiteShockHash,solid,solid,mark=none,line width=0] coordinates { (0.05349,42374.4) (10,42374.4) (10,19990) (0.05349,19990) (0.18186,1.4726e+06) (10,1.4726e+06) (10,19990) (0.18186,19990) (0.05149,26603.1) (10,26603.1) (10,19990) (0.05149,19990) };
          \addplot[mark=flippedTriangle,color=colorBipartiteShockHashFlat,densely dotted,solid,mark=none,line width=0] coordinates { (0.29151,1.4368e+06) (10,1.4368e+06) (10,19990) (0.29151,19990) (0.22504,1.16078e+06) (10,1.16078e+06) (10,19990) (0.22504,19990) (0.182,882044) (10,882044) (10,19990) (0.182,19990) (0.19693,1.04059e+06) (10,1.04059e+06) (10,19990) (0.19693,19990) (0.15897,737115) (10,737115) (10,19990) (0.15897,19990) (0.12207,299642) (10,299642) (10,19990) (0.12207,19990) (0.14094,542859) (10,542859) (10,19990) (0.14094,19990) (0.13436,434316) (10,434316) (10,19990) (0.13436,19990) (0.11611,226302) (10,226302) (10,19990) (0.11611,19990) (0.10424,108588) (10,108588) (10,19990) (0.10424,19990) (0.11106,148403) (10,148403) (10,19990) (0.11106,19990) (0.09821,66668.1) (10,66668.1) (10,19990) (0.09821,19990) (0.09639,45322.3) (10,45322.3) (10,19990) (0.09639,19990) };
          \addplot[mark=phobic,color=colorDensePtHash,densely dotted,solid,mark=none,line width=0] coordinates { (0.62645,2.96209e+06) (10,2.96209e+06) (10,19990) (0.62645,19990) (0.62612,2.94386e+06) (10,2.94386e+06) (10,19990) (0.62612,19990) (0.4955,1.55948e+06) (10,1.55948e+06) (10,19990) (0.4955,19990) (0.49479,1.55807e+06) (10,1.55807e+06) (10,19990) (0.49479,19990) (0.42312,740362) (10,740362) (10,19990) (0.42312,19990) (0.42323,739656) (10,739656) (10,19990) (0.42323,19990) (0.37102,322631) (10,322631) (10,19990) (0.37102,19990) (0.37163,323914) (10,323914) (10,19990) (0.37163,19990) (0.59155,1.5687e+06) (10,1.5687e+06) (10,19990) (0.59155,19990) (0.51722,743793) (10,743793) (10,19990) (0.51722,19990) (0.51389,739645) (10,739645) (10,19990) (0.51389,19990) (0.33529,134320) (10,134320) (10,19990) (0.33529,19990) (0.3355,134897) (10,134897) (10,19990) (0.3355,19990) (0.67678,1.57146e+06) (10,1.57146e+06) (10,19990) (0.67678,19990) (0.5978,739689) (10,739689) (10,19990) (0.5978,19990) (0.59796,743030) (10,743030) (10,19990) (0.59796,19990) (0.67439,744003) (10,744003) (10,19990) (0.67439,19990) (0.67287,740757) (10,740757) (10,19990) (0.67287,19990) (0.45583,324715) (10,324715) (10,19990) (0.45583,19990) (0.42231,134963) (10,134963) (10,19990) (0.42231,19990) (0.4173,134673) (10,134673) (10,19990) (0.4173,19990) (0.53834,324582) (10,324582) (10,19990) (0.53834,19990) (0.60569,322760) (10,322760) (10,19990) (0.60569,19990) (0.55681,134607) (10,134607) (10,19990) (0.55681,19990) (0.48943,134556) (10,134556) (10,19990) (0.48943,19990) (0.48577,134317) (10,134317) (10,19990) (0.48577,19990) (0.68257,324774) (10,324774) (10,19990) (0.68257,19990) (0.62004,134760) (10,134760) (10,19990) (0.62004,19990) };

          \addplot[mark=square,color=colorBipartiteShockHash,solid,color=white,line legend] coordinates { (0.07,40000.0) };
          \addplot[mark=flippedTriangle,color=colorBipartiteShockHashFlat,densely dotted,color=white,line legend] coordinates { (0.2,100000.0) };
          \addplot[mark=phobic,color=colorDensePtHash,densely dotted,color=white,mark=phobicInverse,line legend] coordinates { (0.5,100000.0) };
          \addplot[mark=consensus,color=colorConsensus,color=white,mark=consensusInverse,line legend] coordinates { (0.03,100000.0) };
          
        \end{axis}
    \end{tikzpicture}%
%
	\hfill
	\input{fig/paretoConstructionHeatmapLog2}
	\caption{Dominance maps indicating the approach with the fastest queries, given a specific trade-off between space and construction time with (right) and without (left) MorphisHash on 100\,M keys. Space overhead in bits per key over the lower bound of 1.44.}
	\label{fig:pareto}
\end{figure}

\begin{table}[t]
	\caption{Performance of various methods on 100\,M keys. In the first part, configurations are chosen such that the construction times are about equal (sorted by space efficiency). In the second part, configurations are chosen such that space consumption is almost equal.}
	\label{tab:overviewTable}
	
\addtolength\tabcolsep{-1.9pt}
\small
\centering
\begin{tabular}[t]{l rrr}
	\toprule
	Method & Space & Query & Construction \\
	& (bits/key) & (ns/query) & (ns/key) \\ \midrule

                   Consensus-RS, $k$=$32768$, $o$=$0.0025$ & 1.447 & 222 & 6\,733 \\
   Bip. MorphisHash-RS, base case size $n$=$52$, $b$=n-$4$ & 1.501 & 137 & 6\,669 \\
                Bip. ShockHash-RS, base case size $n$=$66$ & 1.523 & 147 & 7\,186 \\
 Bip. MorphisHash-Flat, base case size $n$=$88$, $b$=n-$2$ & 1.541 &  75 & 6\,330 \\
              Bip. ShockHash-Flat, base case size $n$=$96$ & 1.554 &  76 & 6\,676 \\
                            PHOBIC, $\lambda$=$8.85$, IC-R & 1.749 &  49 & 6\,426 \\
 
 \midrule
 
             Bip. ShockHash-RS, base case size $n$=$128$ & 1.489 & 131 & 172\,738 \\
 Bip. MorphisHash-RS, base case size $n$=$64$, $b$=n-$4$ & 1.489 & 139 &   8\,085 \\

	\bottomrule
	
\end{tabular}

\end{table}

We compare MorphisHash-RS and MorphisHash-Flat with state-of-the-art competitors.
We select the following space efficient competitors based on a recent survey \cite{lehmann2025modern}: \consensus-RS\cite{lehmann2025combined}, ShockHash-Flat, ShockHash-RS and PHOBIC \cite{PHOBIC}.
We test a wide range of configurations for each competitor and compare them in \cref{fig:pareto}.
A selection of configurations is shown in \cref{tab:overviewTable}.
As can be seen in the plot and the table, MorphisHash-RS is about 0.02 bits per key more space efficient than ShockHash-RS when compared for equal construction time.
This corresponds to a reduction of 27\% in space overhead.
When compared for equal space consumption of 1.489 bits per key, MorphisHash-RS is $\frac{172\,738}{8085}\approx \speedup$ times faster to construct (see \cref{tab:overviewTable}).
For MorphisHash-Flat we select $b$ less aggressively compared to MorphisHash-RS (\cref{ss:optibapp}) and obtain a space improvement of about 0.01 bits per key.
According to \cref{fig:pareto} MorphisHash dominates ShockHash in the overall space, construction and query time trade-off.
The next best competitor in terms of space efficiency is PHOBIC which is a clear winner in terms of query throughput.
In the other direction, we have the recently published \consensus-RS, which can reach space overheads as low as 0.001 bits per key at the cost of additional query time.
A negative result regarding non-minimal PHFs can be found \fvref{s:neg-result}.

\section{Conclusion and Future Work}
\label{s:conclusion}
MorphisHash almost completely eliminates the remaining redundancy in ShockHash.
This is particularly effective when combined with a space efficient partitioning technique.
Our compression scheme might be of a more general interest, further examples can be found \fvref{s:appl}.

In future work we plan to improve the space efficiency of partitioning techniques as those are a major source of space overhead.
We are hopeful that a partitioning technique that involves the novel \consensus technique puts further trade-offs for MorphisHash into reach.

\bibliography{paper}

\begin{thebibliography}{10}

\bibitem{belazzougui2010fast}
Djamal Belazzougui, Paolo Boldi, Rasmus Pagh, and Sebastiano Vigna.
\newblock Fast prefix search in little space, with applications.
\newblock In {\em {ESA} {(1)}}, volume 6346 of {\em Lecture Notes in Computer
  Science}, pages 427--438. Springer, 2010.
\newblock \href {https://doi.org/10.1007/978-3-642-15775-2_37}
  {\path{doi:10.1007/978-3-642-15775-2_37}}.

\bibitem{belazzougui2014alphabet}
Djamal Belazzougui and Gonzalo Navarro.
\newblock Alphabet-independent compressed text indexing.
\newblock {\em {ACM} Trans. Algorithms}, 10(4):23:1--23:19, 2014.
\newblock \href {https://doi.org/10.1145/2635816} {\path{doi:10.1145/2635816}}.

\bibitem{bez2022high}
Dominik Bez, Florian Kurpicz, Hans{-}Peter Lehmann, and Peter Sanders.
\newblock High performance construction of recsplit based minimal perfect hash
  functions.
\newblock In {\em {ESA}}, volume 274 of {\em LIPIcs}, pages 19:1--19:16.
  Schloss Dagstuhl -- Leibniz-Zentrum f{\"{u}}r Informatik, 2023.
\newblock \href {https://doi.org/10.4230/LIPICS.ESA.2023.19}
  {\path{doi:10.4230/LIPICS.ESA.2023.19}}.

\bibitem{bourgain2010singularity}
Jean Bourgain, Van~H Vu, and Philip~Matchett Wood.
\newblock On the singularity probability of discrete random matrices.
\newblock {\em Journal of Functional Analysis}, 258(2):559--603, 2010.

\bibitem{broder2004network}
Andrei~Z. Broder and Michael Mitzenmacher.
\newblock Survey: Network applications of {B}loom filters: {A} survey.
\newblock {\em Internet Math.}, 1(4):485--509, 2003.
\newblock \href {https://doi.org/10.1080/15427951.2004.10129096}
  {\path{doi:10.1080/15427951.2004.10129096}}.

\bibitem{chang2005perfect}
Chin{-}Chen Chang and Chih{-}Yang Lin.
\newblock Perfect hashing schemes for mining association rules.
\newblock {\em Comput. J.}, 48(2):168--179, 2005.
\newblock \href {https://doi.org/10.1093/COMJNL/BXH074}
  {\path{doi:10.1093/COMJNL/BXH074}}.

\bibitem{crawford2018practical}
Victoria~G. Crawford, Alan Kuhnle, Christina Boucher, Rayan Chikhi, and Travis
  Gagie.
\newblock Practical dynamic de bruijn graphs.
\newblock {\em Bioinform.}, 34(24):4189--4195, 2018.
\newblock \href {https://doi.org/10.1093/BIOINFORMATICS/BTY500}
  {\path{doi:10.1093/BIOINFORMATICS/BTY500}}.

\bibitem{dillinger2021fast}
Peter~C Dillinger, Lorenz H{\"u}bschle-Schneider, Peter Sanders, and Stefan
  Walzer.
\newblock Fast succinct retrieval and approximate membership using ribbon.
\newblock {\em arXiv preprint arXiv:2109.01892}, 2021.

\bibitem{esposito2020recsplit}
Emmanuel Esposito, Thomas~Mueller Graf, and Sebastiano Vigna.
\newblock {RecSplit}: Minimal perfect hashing via recursive splitting.
\newblock In {\em {ALENEX}}, pages 175--185. {SIAM}, 2020.
\newblock \href {https://doi.org/10.1137/1.9781611976007.14}
  {\path{doi:10.1137/1.9781611976007.14}}.

\bibitem{sourceCodeMorphisHash}
Stefan Hermann.
\newblock {MorphisHash - GitHub}.
\newblock \url{https://github.com/stefanfred/MorphisHash}, 2025.

\bibitem{PHOBIC}
Stefan Hermann, Hans-Peter Lehmann, Giulio~Ermanno Pibiri, Peter Sanders, and
  Stefan Walzer.
\newblock {PHOBIC:} perfect hashing with optimized bucket sizes and interleaved
  coding.
\newblock In {\em {ESA}}, volume 308 of {\em LIPIcs}, pages 69:1--69:17.
  Schloss Dagstuhl - Leibniz-Zentrum f{\"{u}}r Informatik, 2024.
\newblock \href {https://doi.org/10.4230/LIPIcs.ESA.2024.69}
  {\path{doi:10.4230/LIPIcs.ESA.2024.69}}.

\bibitem{sourceCodeMphfExperiments}
Hans-Peter Lehmann.
\newblock {MPHF Experiments - GitHub}.
\newblock \url{https://github.com/ByteHamster/MPHF-Experiments}, 2025.

\bibitem{lehmann2025modern}
Hans-Peter Lehmann, Thomas Mueller, Rasmus Pagh, Giulio~Ermanno Pibiri, Peter
  Sanders, Sebastiano Vigna, and Stefan Walzer.
\newblock Modern minimal perfect hashing: A survey.
\newblock {\em arXiv preprint arXiv:2506.06536}, 2025.

\bibitem{lehmann2023bipartite}
Hans-Peter Lehmann, Peter Sanders, and Stefan Walzer.
\newblock {ShockHash}: Near optimal-space minimal perfect hashing beyond
  brute-force.
\newblock {\em arXiv preprint, invited to Algorithmica}, 2024.
\newblock \href {https://doi.org/10.48550/ARXIV.2310.14959}
  {\path{doi:10.48550/ARXIV.2310.14959}}.

\bibitem{lehmann2023shockhash}
Hans{-}Peter Lehmann, Peter Sanders, and Stefan Walzer.
\newblock Shockhash: Towards optimal-space minimal perfect hashing beyond
  brute-force.
\newblock In {\em {ALENEX}}. {SIAM}, 2024.
\newblock \href {https://doi.org/10.1137/1.9781611977929.15}
  {\path{doi:10.1137/1.9781611977929.15}}.

\bibitem{lehmann2025combined}
Hans-Peter Lehmann, Peter Sanders, Stefan Walzer, and Jonatan Ziegler.
\newblock Combined search and encoding for seeds, with an application to
  minimal perfect hashing.
\newblock In {\em {ESA}}, LIPIcs. Schloss Dagstuhl - Leibniz-Zentrum f{\"{u}}r
  Informatik, 2025.

\bibitem{lu2006perfect}
Yi~Lu, Balaji Prabhakar, and Flavio Bonomi.
\newblock Perfect hashing for network applications.
\newblock In {\em {ISIT}}, pages 2774--2778. {IEEE}, 2006.
\newblock \href {https://doi.org/10.1109/ISIT.2006.261567}
  {\path{doi:10.1109/ISIT.2006.261567}}.

\bibitem{mehlhorn1982program}
Kurt Mehlhorn.
\newblock On the program size of perfect and universal hash functions.
\newblock In {\em {FOCS}}, pages 170--175. {IEEE} Computer Society, 1982.
\newblock \href {https://doi.org/10.1109/SFCS.1982.80}
  {\path{doi:10.1109/SFCS.1982.80}}.

\bibitem{newman2010networks}
MARK~EJ Newman.
\newblock Networks: an introduction, 2010.

\bibitem{pagh2008uniform}
Anna Pagh and Rasmus Pagh.
\newblock Uniform hashing in constant time and optimal space.
\newblock {\em SIAM Journal on Computing}, 38(1):85--96, 2008.

\bibitem{pagh2004cuckoo}
Rasmus Pagh and Flemming~Friche Rodler.
\newblock Cuckoo hashing.
\newblock {\em Journal of Algorithms}, 51(2):122--144, 2004.

\bibitem{pibiri2022sparse}
Giulio~Ermanno Pibiri.
\newblock Sparse and skew hashing of k-mers.
\newblock {\em Bioinformatics}, 38(Supplement\_1):i185--i194, 2022.

\bibitem{pibiri2017efficient}
Giulio~Ermanno Pibiri and Rossano Venturini.
\newblock Efficient data structures for massive \emph{N}-gram datasets.
\newblock In {\em {SIGIR}}, pages 615--624. {ACM}, 2017.
\newblock \href {https://doi.org/10.1145/3077136.3080798}
  {\path{doi:10.1145/3077136.3080798}}.

\bibitem{takacs1990cayley}
Lajos Tak{\'a}cs.
\newblock On cayley's formula for counting forests.
\newblock {\em Journal of Combinatorial Theory, Series A}, 53(2):321--323,
  1990.

\end{thebibliography}
\clearpage
\appendix

\section{Other Compression Applications}
\label{s:appl}
In this section we give further examples where our technique of combining linear constraints with a retrieval structure can be applied.

\subsection{Compressing a 2-Coloring}
\myparagraph{Problem Statement.}
Given a graph $(V, E)$ with $|V|=n$ vertices.
A 2-coloring is a function $f: V \rightarrow\mathbb{F}_2$ such that for all edges $\{v, w\} \in E: f(v) \neq f(w)$.
The problem is to store a 2-coloring and allow querying of $f$ without knowing the set of edges.
We assume that the graph has at least one 2-coloring.

\myparagraph{The Affine Solution Space.}
The idea is that for each component of the graph we are allowed to flip the coloring of the respective nodes.
Given any 2-coloring $y \in \mathbb{F}_2^V$ such that $y_i$ is the color of node $i \in V$, we may flip the color of a component by adding $z \in \mathbb{F}_2^V$ where $z_i=1$ if and only if the node is in the component.
We can do so for each component independently resulting in an affine subspace.

\myparagraph{The Equation System.}
We are now interested in the equation system which has this affine subspace as solution space.
We can describe the constraints of a 2-coloring using linear equations.
For each edge $\{v, w\} \in E$ we have the constraint that $y_v \neq y_w$ which is equivalent to $y_v + y_w = 1$.
Let $A \in \mathbb{F}_2^{E \times V}$ be the transposed incidence matrix of the graph.
Clearly, the system we have to solve is $Ay=\mathbf{1}$.
Like in MorphisHash we combine the system with a retrieval structure using a random matrix $H \in \mathbb{F}_2^{V \times b}$ to obtain $AHx=\mathbf{1}$, where $b$ is again the tuning parameter.

\myparagraph{Analysis.}
Let $c$ be the number of components of the graph.
Clearly, the lower bound for storing a 2-coloring is $n-c$ bits because there are two ways to color each component.
Choosing $b=n-c$ and using an expected constant number of retries for the random matrix if the system has no solution we can store the coloring using an expected $n-c + \mathcal{O}(1)$ bits or $n-c + \mathcal{O}(\log(c))$ bits if we also have to store $c$ itself.

\subsection{A Retrieval Structure that Only Allows Difference Queries}
\myparagraph{Problem Statement.}
Given a set of $n$ keys $S$ taken from a large universe $U$ and a function $f: S \rightarrow [p]$, where $p$ is a prime number.
We are looking for a data structure that gives the result of $f(a) - f(b)$ for any pair $a, b \in S$.
The data structure knows $S$ and $f$ during construction but not afterwards.

\myparagraph{The Affine Solution Space.}
Let $y\in\mathbb{F}_p^S$ be such that $y_i = f(i)$.
Clearly, we can add any multiple of $\mathbf{1}$ to $y$ while maintaining the differences $y_a-y_b$.

\myparagraph{The Equation System.}
We assign an arbitrary order to $S$.
The difference of consecutive keys $a$ and $b$ result in $n-1$ linear constraints $y_a - y_b = f(a) - f(b)$.
As always, we combine this system with a retrieval structure $y=Hx$ using the solution vector $x \in \mathbb{F}_p^b$ and a random matrix $H \in \mathbb{F}_p^{n \times b}$.

\myparagraph{Analysis.}
Clearly, the space lower bound for this problem is $\log_2(\frac{p^n}{p}) = (n-1) \log_2(p)$ bits.
Using $b=p-1$, our approach requires an expected constant number of retries and therefore $(n-1) \log_2(p) + \mathcal{O}(1)$ bits in expectation when storing $x$ using arithmetic coding.

\section{Negative Result for Non-Minimal PHF}
\label{s:neg-result}
We also considered our compression technique for the non-minimal case.
In the non-minimal case, each component either has one more node then edges or the same number of nodes and edges.
A component with the same number of nodes and edges has two possible orientations of the cycle and all other edges are oriented away from the cycle.
In a component with $t$ nodes and $t-1$ edges (i.e. a tree) any of the $t$ nodes can be selected as roots and all edges are oriented away from that root.
However, there are two reasons why we cannot efficiently compress in the non-minimal case.

\begin{enumerate}
	\item Our technique is not applicable because the possible orientations of a tree with $t>2$ cannot be expressed as the solution of a linear equation system and would therefore require non-linear solving techniques.
	
	\item We found that the entropy of the edge orientations is higher than the space lower bound for non-minimal PHFs.
	For a tree with $t_i$ nodes and therefore $t_i-1$ edges, $t_i$ out of $2^{t_i - 1}$ orientations of the tree are possible.
	The entropy therefore calculates as $\sum \log_2(2^{t_i - 1}/t_i) = \sum t_i - 1 - \log_2(t_i)$ bits.
	For trees with an additional edge we charge $\sum p_i - 1$ bits, where $p_i$ is the number of nodes.
	For a load of 50\%, the PHF lower bound is 0.44 bits per key and the experimentally determined entropy of edge orientations is 0.56 bits per key.
	
\end{enumerate}

\begin{figure}[h!]
	    \begin{tikzpicture}
        \begin{axis}[
            plotEvalPareto,
            xlabel={Overhead (Bits/Key)},
            ylabel={Throughput (Keys/s)},
            xmode=log,
            ymode=log,
            xmin=0.01,
            xmax=1,
            ymin=1.01e4,
            log x ticks with fixed point,
            legend to name=legendEvalParetoQuery,
            legend style={nodes={scale=0.9, transform shape}},
            legend columns=3,
            only marks
          ]

          \tikz \fill [white] (10000.1,0) rectangle (00.2,0.2);
          
          \addplot[mark=flippedTriangle,color=colorBipartiteShockHashFlat,densely dotted,mark repeat*=4] coordinates { (0.09639,45322.3) (0.09821,66668.1) (0.10424,108588) (0.11106,148403) (0.11611,226302) (0.12207,299642) (0.13436,434316) (0.14094,542859) (0.15897,737115) (0.182,882044) (0.19693,1.04059e+06) (0.22504,1.16078e+06) (0.29151,1.4368e+06) (0.4083,1.54887e+06) };
          \addlegendentry{Bip. ShockHash-Flat \cite{lehmann2023bipartite}}
          \addplot[mark=square,color=colorBipartiteShockHash,solid,mark repeat*=4] coordinates { (0.04647,5783) (0.04809,13086.9) (0.05149,26603.1) (0.05349,42374.4) (0.05751,58515.4) (0.06039,78145) (0.0654,80112.2) (0.06684,94599.7) (0.07216,95551) (0.07256,115576) (0.07861,127014) (0.08194,174442) (0.08718,178979) (0.08879,180093) (0.09349,180178) (0.09743,205953) (0.11035,214581) (0.11541,218186) (0.12625,256523) (0.13574,265258) (0.14165,273286) (0.15256,279853) (0.1753,415828) (0.18125,447417) (0.18186,1.4726e+06) (0.1874,1.86411e+06) (0.20967,2.01983e+06) };
          \addlegendentry{Bip. ShockHash-RS \cite{lehmann2023bipartite}}
          \addplot[mark=phobic,color=colorDensePtHash,densely dotted] coordinates { (0.33529,134320) (0.3355,134897) (0.37102,322631) (0.37163,323914) (0.42312,740362) (0.49479,1.55807e+06) (0.4955,1.55948e+06) (0.59155,1.5687e+06) (0.62612,2.94386e+06) (0.62645,2.96209e+06) (0.84162,2.97221e+06) (0.89763,4.69484e+06) (0.89795,4.6966e+06) (1.0587,4.70323e+06) (1.06079,4.73126e+06) (1.17768,4.74068e+06) (1.4084,6.28694e+06) (1.40996,6.30835e+06) (1.66198,6.32671e+06) (1.82962,6.35647e+06) (2.014,6.36497e+06) (2.01759,6.38488e+06) (2.2887,6.39672e+06) (2.32278,7.32493e+06) (2.67008,7.36269e+06) (2.95334,7.3888e+06) (3.36902,7.44879e+06) (3.92697,7.47217e+06) (3.93516,7.48447e+06) (4.59104,7.51597e+06) (4.59864,7.52219e+06) };
          \addlegendentry{PHOBIC \cite{PHOBIC}}
          \addplot[mark=Mercedes star flipped,color=colorMorphisHashFlat,mark repeat*=4] coordinates { (0.09243,47300) (0.09393,102908) (0.0959,103086) (0.09744,127945) (0.09865,157620) (0.10855,193239) (0.10919,227111) (0.11489,280389) (0.11903,318856) (0.12888,393622) (0.13318,429934) (0.13843,430929) (0.14098,529574) (0.15037,604336) (0.16163,732354) (0.16456,738465) (0.18762,860867) (0.22167,1.02713e+06) (0.25764,1.18912e+06) (0.26932,1.19096e+06) (0.31269,1.32709e+06) (0.32991,1.32853e+06) (0.38885,1.46175e+06) (0.38894,1.46671e+06) (0.4966,1.5015e+06) (0.51005,1.50188e+06) };
          \addlegendentry{\textbf{Bip. MorphisHash-Flat} }
          \addplot[mark=star,color=colorMorphisHash,mark repeat*=4] coordinates { (0.03133,32634.6) (0.03307,44738.4) (0.03595,55523.3) (0.03655,56569.8) (0.03835,67518.2) (0.03835,71237.4) (0.04195,85422.8) (0.04313,85752.7) (0.04577,99457.9) (0.04603,123230) (0.0501,143741) (0.05519,146927) (0.0574,157697) (0.06254,161635) (0.06319,167978) (0.06856,179625) (0.06881,182790) (0.07197,183343) (0.07394,199917) (0.07914,209299) (0.07961,213270) (0.08509,218695) (0.08812,227812) (0.08813,236607) (0.0936,243769) (0.09733,244995) (0.09803,260805) (0.10422,275736) (0.10591,288264) (0.11246,308062) (0.12141,330516) (0.12165,331601) (0.12672,343020) (0.13616,366932) (0.13709,386559) (0.13954,730807) (0.14747,838771) (0.15105,1.03358e+06) (0.15631,1.17655e+06) (0.17977,1.32899e+06) (0.19196,1.43976e+06) };
          \addlegendentry{\textbf{Bip. MorphisHash-RS}}
          \addplot[mark=consensus,color=colorConsensus] coordinates { (0.011,29714.9) (0.011,59296) (0.011,170538) (0.011,325188) (0.01505,502748) (0.0508,1.54756e+06) (0.08762,2.62343e+06) (0.13553,3.72842e+06) (0.60185,6.67067e+06) };
          \addlegendentry{\consensus-RS \cite{lehmann2025combined}}

        \end{axis}
    \end{tikzpicture}
\end{figure}

\FloatBarrier
\end{document}